\documentclass{article}

\usepackage{custom}
\usepackage{stat}
\usepackage{forArxiv}

\begin{document}

\newtheorem{lemme}{Lemma}
\newtheorem{prop}{Proposition}
\newtheorem{thm}{Theorem}
\newtheorem{corollary}{Corollary}

\author{Olivier Peltre}

\date{Université Denis Diderot, Paris, France}


\title{A Homological Approach to Belief Propagation and Bethe Approximations}
\maketitle

\begin{abstract}

We introduce a differential complex of 
local observables given a decomposition of a global set of random variables
into subsets.
Its boundary operator $\bord$ allows us to define
a transport equation $\dot u = \dr \Phi(u)$ 
equivalent to Belief Propagation.
This definition reveals a set
of conserved quantities under Belief Propagation
and gives new insight on the relationship of its equilibria with 
the critical points of Bethe free energy. 

\end{abstract}

\section*{Introduction}

\newcommand{\cobord}[1]{\delta\Lambda^{#1}}

A common feature of statistical physics and statistical
learning is to consider a very large number of random variables,
each of them mostly interacting
with only a small subset of neighbours.
Both lead to the effort of extracting relevant information about 
collective phenomena in spite of intractable global computations,
hence motivating the development of local techniques where only small 
enough subsets of variables are simultaneously considered.

In the present note, 
we substitute the space of global observables with a collection of local ones,
where relations between intersecting subsets of variables 
are translated into a boundary operator. 
Our construction yields a differential complex which we expose in section 1.
This allows for a homological interpretation
of some well-known techniques having motivated the present work,
its starting point being the equivalence established by Yedidia {\it et. al.}
between critical points of the Bethe free energy 
and fixed points of the Belief Propagation algorithm.
We review this beautiful theorem bridging statistical learning
and thermodynamics in section 2.

Bethe first proposed in 1935
a combinatorial method to compute an approximate value of the free energy 
per atom in a bipartite cristal, expanding computations over concentric
shells around a given lattice point \cite{Bethe-35}.
His method was soon generalised by Peierls \cite{Peierls-36}
and others, but its most important refinement was brought in 1951 by
Kikuchi's cluster variation method \cite{Kikuchi-51}. 
These various methods solve a common combinatorial problem, 
first related by Morita to the general theory of Möbius inversion \cite{Morita-57}
and we generically call them Bethe approximations.

In the following we denote by $\Om$ a finite set 
of variables, and by $X$ a collection $\aa, \bb, \cc$,... of subsets 
of $\Om$. 
The approximate value $\Fb$ of the global Gibbs free energy $\Fg_\Om$ is
computed as a weighted sum of each cluster's Gibbs free energy:
\[ \Fg_\Om \quad  \simeq  \quad \Fb = \sum_{\aa} c_\aa \Fg_\aa \]
where each $c_\aa$ is a suitably chosen integer to account for the redundancies
brought by intersecting clusters.
Through a variational principle on $\Fb$, 
one may try to approximate the marginals of the Gibbs state on each cluster.

Belief Propagation is a bayesian inference algorithm introduced by Pearl in 1988
\cite{Pearl-88}, equivalent to the low density decoding algorithms 
defined by Gallager in 1963 \cite{Gallager-63}. 
It was initially considered on graphs and
known to converge to the exact marginal probabilities on trees, 
while graphs with loops show degeneracy of its fixed points
\cite{Weiss-97,Murphy-Weiss-99}.  
The generalised version of Yedidia {\it et. al.} 
considers more general clusters $\aa \incl \Om$ \cite{Yedidia-2001,Yedidia-2005}.
Allowing for a better trade-off between complexity and accuracy,
it describes one of the widest class of learning algorithms 
.

Let us denote by $\Lambda^\aa$ the set of subclusters of $\aa$, 
and by $\cobord{\aa}$ the set of ordered pairs $\aa' \cont \bb'$ such that 
$\aa' \not\in \Lambda^\aa$ and $\bb' \in \Lambda^\aa$.
The algorithm defines a dynamic over a set of beliefs 
$(q_\aa)$ and a set of messages $(m_{\aa\bb})$,
where for $\aa$ strictly containing $\bb$, 
each $m_{\aa\bb}$ is a strictly positive observable on $\bb$ 
describing a message 
from $\Lambda^\aa \setminus \Lambda^\bb$ to $\Lambda^\bb$.
Given a collection of strictly positive functions $(f_\bb)$ 
accounting for priors and local evidence, 
each belief is the non-vanishing Markov field 
defined as:
\[ q_\aa = \Bigg[ 
\prod_{\bb' \in \Lambda^\aa} f_{\bb'} 
\; \times 
\prod_{\aa'\bb' \in \cobord{\aa}} m_{\aa'\bb'} 
\Bigg] \] 
where the bracket denotes normalisation.
Denoting by $\Sigma^{\bb\aa}$ the integration over the 
variables in $\aa \setminus \bb$, 
messages follow the update rule:
\[ m_{\aa\bb} \wa m_{\aa\bb} \cdot 
\Sigma^{\bb\aa} \left( 
\frac{
    \prod\limits_{\bb' \in \Lambda^\aa \setminus \Lambda^\bb} f_{\bb'} 
    \; \times
    \prod\limits_{\aa'\bb' \in \cobord{\aa} \setminus \cobord{\bb}} m_{\aa'\bb'}
}{
    \prod\limits_{\bb'\cc' \in \cobord{\bb} \setminus \cobord{\aa}} m_{\bb'\cc'} 
}
\right)
\]
which is sometimes called a sum-product algorithm.
Substituting beliefs, the update rule may take the much nicer form:
\[ m_{\aa\bb} \wa m_{\aa\bb} \cdot \frac {\Sigma^{\bb\aa}(q_\aa)} {q_\bb} \]
Belief Propagation essentially searches for a consistent statistical field $(q_\aa)$,
{\it i.e.} such that $q_\bb$ is the marginal of $q_\aa$ 
whenever $\bb$ is contained in $\aa$,
inside a particular subspace defined by the priors $(f_\aa)$.
This subspace can be characterised by a maximal set of conserved quantities
under a transport equation of the form $\dot u = \bord \Phi(u)$ 
equivalent to Belief Propagation, where $\bord$ is the boundary operator
in degree zero of a natural homology theory. 

This is part of a PhD work done under the kind supervision of Daniel Bennequin. 
I wish to thank him as well as Grégoire Sergeant-Perthuis 
and Juan-Pablo Vigneaux, for their sustained collaboration and 
many fruitful discussions.

\section{Differential and Combinatorial Structures}
\subsection{Statistical System}

\subsubsection{Regions.}

We will call system a finite set $\Omega$ equipped with a 
collection of subsets $X \incl \Part(\Omega)$ such that: 
\begin{itemize}[noitemsep,topsep=-6pt]
\iii the empty set $\vide$ is in $X$,
\iii if $\aa \in X$ and $\bb \in X$, then $\aa \cap \bb$ is also in $X$.
\end{itemize}
We say that $X$ is an admissible decomposition\footnote{
    A general $X \incl \Part(\Omega)$ is called a hypergraph. 
    These two axioms are necessary for the interaction decomposition
    (theorem 1) to hold.
} of $\Omega$
when it satisfies the above axioms. 
We call $\aa \in X$ a region\footnotemark{} of $\Omega$ 
and denote by $\Lambda^\aa \incl X$ the subsystem of 
those regions contained in $\aa$.

\footnotetext{
    The term is chosen to refer to the notion of region-graphs
    in Yedidia et al. 
}

We may view $X$ as a
subcategory of the partial order $\Part(\Omega)$
having an arrow $\aa \aw \bb$ whenever $\bb \incl \aa$.
The above axioms state that $\vide$ is a final element of $X$ 
and that $X$ has sums.

\subsubsection{Chains and Nerve.}

A $p$-chain $\bar{\aa}$
is a totally ordered sequence $\aa_0 \aw \dots \aw \aa_p$ in $X$.
It is said non-degenerate when all inclusions are strict and
we write $\bar{\aa} = \aa_0 > \dots > \aa_p$ 
to emphasize on non-degeneracy.
A $p$-chain $\bar{\aa}$ may be viewed as a $p$-simplex.
Its $p+1$ faces are the chains $(\dr_k \bar\aa)_{0 \leq k \leq p}$ 
of degree $p-1$ obtained by removing $\aa_k$. 

The nerve of $X$ is the simplicial complex\footnote{
    For the general theory of simplicial complexes, see \cite{Eilenberg-Steenrod}.
}
$NX = \bigcup_p N^p X$ formed by all non-degenerate chains.
Its dimension is the degree of the longest non-degenerate chain in $X$,
bounded by the cardinal of $\Omega$. 

\subsubsection{Microscopic States.}

For each $i \in \Omega$, suppose given a finite set $E_i$.
A microscopic state of a region $\aa \incl \Omega$ is an element of the 
cartesian product\footnote{
The configuration space $E_\vide$ is thus a point, 
unit for the cartesian product and terminal element in $\Set$.
}:
\[ E_\aa = \prod_{i \in \aa} E_i \]
An element $x_\aa \in E_\aa$ is also called a configuration of $\aa$.
Configuration spaces yield a functor from $\Part(\Omega)$ to $\Set$ 
and we denote by $\pi^{\bb\aa} : E_\aa \aw E_\bb$ the canonical projection
of $E_\aa$ onto $E_\bb$ whenever $\bb$ is a subregion of $\aa$.

In the following, we will be mostly interested in 
the restriction of such a functor
$E : \Part(\Omega) \aw \Set$ 
to an admissible decomposition $X \incl \Part(\Omega)$.

\subsection{Scalar Fields}

\subsubsection{Differentials.}

We call scalar field a collection $\lambda \in \R(X)$ 
of scalars indexed by the nerve of $X$. 
We denote by $\R_p(X)$ the space of $p$-fields, vanishing everywhere but 
on the $p$-simplices of $NX$:
\[ \R_p(X) = \R^{N^p X} \]

Using the canonical scalar product of $\R(X)$,
scalar fields can be identified with chains or cochains 
with real coefficients in $NX$ eitherwise.
We denote by $\bord$ the boundary operator of $\R(X)$ and
by $d$ its adjoint differential\footnote{
    We shall give formulas for these operators shortly. 
    A nice presentation can be found in Kodaira \cite{Kodaira-49}
}:
\[ \begin{array}{cc}
\bord : & \R_0(X) \wa \R_1(X) \wa \dots \\
d :     & \R_0(X) \aw \R_1(X) \aw \dots 
\end{array} \]
Because $X$ has $\vide$ as final element, 
$NX$ is a retractable simplicial complex and its homology is acyclic.

\subsubsection{Convolution.}

We denote by $\tilde\R_1(X) = \R_0(X) \oplus \R_1(X)$ the space of
scalar fields of degree lower than one.
Equivalently, an element of $\tilde\R_1(X)$
is indexed by general $1$-chains in $X$, 
identifying the degenerate $1$-chain $\aa \aw \aa$ with $\aa$.

Equipped with Dirichlet convolution,
$\tilde\R_1(X)$ is the incidence algebra\footnote{
    See Rota \cite{Rota-64} for a deeper treatment of these 
    combinatorial structures.
} of $X$.
The product of $\ph$ and $\psi \in \tilde\R_1(X)$ 
is defined for every $\aa \aw \cc$ by:
\[ (\ph * \psi)_{\aa\cc} = \sum_{\aa \aw \bb' \aw \cc}
\ph_{\aa\bb'} \cdot \psi_{\bb'\cc} \]
The unit of $*$ is $1 \in \R_0(X)$, sometimes viewed
as a Kronecker symbol in $\tilde\R_1(X)$.

The incidence algebra of $X$ has two natural actions on the $0$-fields.
They give $\R_0(X)$ a $\tilde\R_1(X)$ bi-module structure. 
The left action of $\ph \in \tilde\R_1(X)$ on $\lambda \in \R_0(X)$ is given by:
\[ (\ph \cdot \lambda)_\aa = \sum_{\aa \aw \bb'} \ph_{\aa\bb'} \lambda_{\bb'} \] 
We will sometimes denote by $\ph$ and $\ph^*$ the two adjoint 
endomorphisms of $\R_0(X)$
associated to the left and right action of $\ph$ respectively.

\subsubsection{Möbius Inversion.}

We denote by $\zeta \in \tilde\R_1(X)$ the Dirichlet zeta function 
defined by $\zeta_{\aa\bb} = 1$ for every $\aa \aw \bb$ in $X$.
The $k$-th power of $\zeta$ 
counts chains of length $k$ 
and $(\zeta - 1)^{*k}_{\aa\bb}$ is the number of non-degenerate $k$-chains
from $\aa$ to $\bb$. 

Because $X$ is locally finite\footnote{
    $X$ is locally finite if for any $\aa,\bb \in X$ 
    there is only a finite number of non-degenerate chains from $\aa$ to $\bb$.
}, 
we easily recover the fundamental theorem of Möbius inversion
stating that $\zeta$ is invertible. Its inverse $\mu$ satisfies:
\[ \mu_{\aa \bb} = \sum_{k \geq 0} (-1)^k (\zeta - 1)^{*k}_{\aa\bb} \]
The field $\mu \in \tilde\R_1(X)$ is known as the Möbius function 
and is valued in $\Z$. 

We also introduce the right Möbius numbers as the coefficients $c \in \R_0(X)$ 
defined by:
\[ c_\bb = \sum_{\aa' \aw \bb} \mu_{\aa'\bb} = (1 \cdot \mu)_\bb \] 
These integers satisfy the following <<inclusion-exclusion>> formula:
\[ (c \cdot \zeta)_\bb = \sum_{\aa' \aw \bb} c_\bb = 1 \]

\subsection{Observables, Densities and Statistical States.}

\subsubsection{Observable Fields.}

An observable on a region $\aa \incl \Om$ is a real function 
on $E_\aa$, we denote by $\obs_\aa = \R^{E_\aa}$ the commutative algebra
of observables on $\aa$. 
For every subregion $\bb \incl \aa$, any observable $u_\bb \in \obs_\bb$
admits a cylindrical extension,
$j_{\aa\bb}(u_\bb) \in \obs_\aa$ defined by pulling back
the projection of $E_\aa$ onto $E_\bb$:
\[ j_{\aa\bb}(u_\bb) = u_\bb \circ \pi^{\bb\aa} \]
The injections $j_{\aa\bb}$ make $\obs:X^{op} \aw \Alg$
a contravariant functor to the category of unital real algebras.
We denote by $1_\aa$ the unit of $\obs_\aa$ and 
will generally drop the injection in our notation.  

We call observable field an allocation of observables on the nerve of $X$,
defining an observable on a $p$-chain $\bar\aa \in N^p X$ 
to be an observable on its smallest region $\aa_p$. 
We write $\obs_{\bar \aa}$ for the copy of $\obs_{\aa_p}$ on  $\bar \aa$,
and denote the space of $p$-fields by:
\[ \obs_p(X) = \bigoplus_{\bar \aa \in N^p X} \obs_{\bar \aa} \]
and by $\obs(X)$ the graded vector space of observable fields.
If $\bar \bb$ is a face of $\bar \aa$, its terminal region contains that
of $\bar \aa$. 
It follows that there is a canonical injection 
$j_{\bar \bb \bar \aa} : \obs_{\bar \aa} \aw \obs_{\bar \bb}$ whenever
$\bar \bb$ is a subchain of $\bar \aa$.\footnote{
    Observable fields form a simplicial algebra 
    $\obs(X) : \Ord^{op} \aw \Alg$. 
    Our construction is related to this more 
    general theory,
    see Segal's note on classifying spaces \cite{Segal-68} for instance.
}

The boundary operator $\bord : \obs_{p+1}(X) \aw \obs_p(X)$
 is defined on a $(p+1)$-field $\ph$ by the following formula:
\[ (\bord \ph)_{\bar \bb}  = \sum_{k=0}^{p+1} \,
\sum_{\bord_k \bar \aa' = \bar \bb} (-1)^k \, j_{\bar \bb \bar \aa'}(\ph_{\bar \aa'}) \]
The proof of $\bord^2 = 0$ is done as usual, 
it uses functoriality of the injections but otherwise does not differ from
the case of scalar coefficients.
In the case where $p = 0$, we have for example:
\[ \bord_\bb \ph = \sum_{\aa' \aw \bb} \ph_{\aa'\bb} 
- \sum_{\bb \aw \cc'} \ph_{\bb \cc'} \]

Each $\obs_\aa$ 
may also be viewed as the abelian Lie algebra of the multiplicative
Lie group 
$G_\aa = (\R_+^*)^{E_\aa}$
of strictly positive observables on $\aa$.
We denote by: 
\[ \expm_\aa : \obs_\aa \aw G_\aa \]
the group isomorphism $u_\aa \mapsto \e^{-u_\aa}$ 
and by $\mlog_\aa$ its inverse.

\subsubsection{Density Fields.}

We call density on a region $\aa \incl \Om$ a linear form on observables 
$\om_\aa \in \obs^*_\aa$. For every $\bb \incl \aa$, we denote 
by $\Sigma^{\bb\aa}(\om_\aa) \in \obs^*_\bb$ the density defined 
by integrating $\om_\aa$ along the fibers of $\pi^{\bb\aa}$:
\[ \Sigma^{\bb \aa}(\omega_\aa)(x_\bb) = 
\sum_{x' \in E_{\aa \setminus \bb}} \omega_\aa(x_\bb, x') \]
The marginal projections $\Sigma^{\bb\aa}$ make $\obs^* : X \aw \Vect$
a covariant functor. 
Each $\Sigma^{\bb\aa}$ is
dual to $j_{\aa\bb}$ as 
for every $\omega_\aa \in \obs^*_\aa$ and every $u_\bb \in \obs_\bb$, 
we have:
\[ \croc{\Sigma^{\bb \aa}(\omega_\aa)}{u_\bb} 
= \croc{\omega_\aa}{j_{\aa\bb}(u_\bb)} \]

We define the complex $\obs^*(X)$ of density fields
by duality with $\obs(X)$, 
and denote its differential by $\diff$, adjoint of $\bord$. 
The action of $\diff$ on a density $p$-field $\omega$ 
is the $(p+1)$-field defined by:
\[ (\diff \omega)_{\bar \aa} = \sum_{k = 0}^p (-1)^k \, 
\Sigma^{\bar \aa \, \bord_k \bar \aa}(\omega_{\bord_k \bar \aa}) \]
writing $\Sigma^{\bar \aa \bar \bb}$ for the adjoint of $j_{\bar \bb \bar \aa}$.
For $p = 0$, we have:
\[ (\diff \omega)_{\aa \bb} = \omega_\bb - \Sigma^{\bb \aa}(\omega_\aa) \]
We say that a field $\omega \in \obs^*_0(X)$ is consistent if 
$d\om = 0$. 
For every region $\aa \in X$ and
every subregion $\bb \incl \aa$, the marginal
of $\omega_\aa$ on $\bb$ is then $\omega_\bb$.

A density $\om_\aa$ on $\aa$ is said positive if $\om_\aa(u^2_\aa) \geq 0$
for every observable $u_\aa$.
The counting measure on $E_\aa$, denoted by $1^\flat_\aa$,
is positive. Its free orbit under multiplication by the observables of $G_\aa$
is the space of strictly positive densities, 
we will sometimes identify $G_\aa$ with $G_\aa \cdot 1^\flat_\aa \incl \obs^*_\aa$.

\subsubsection{Statistical Fields.}

A statistical state on a region $\aa \incl \Om$ is a positive 
density $p_\aa \in \obs^*_\aa$ such that $p_\aa(1_\aa) = 1$.
Equivalently, it is a probability measure on $E_\aa$ and 
we denote by $\Delta_\aa = \Prob(E_\aa)$ the convex set of 
statistical states on $\aa$.
It is a topological simplex of dimension $|E_\aa| - 1$ 
and its image under $\Sigma^{\bb\aa}$ is $\Delta_\bb$, so that
$\Delta : X \aw \Top$ defines a covariant functor of topological spaces.
Any non-trivial positive density $\om_\aa \in \obs^*_\aa$ defines a
normalised statistical state which we denote by $[\om_\aa] \in \Delta_\aa$.

The interior $\Del_\aa$ of the probability simplex on $E_\aa$
consists of all the non-vanishing probability densities. 
It is diffeomorphic to the quotient of $G_\aa$ by scalings of $\R_+^*$,
itself isomorphic to the quotient of $\obs_\aa$ by the action of additive constants:
\[ \Del_\aa \sim \P G_\aa \sim \obs_\aa / \R \]
It follows that $\Del_\aa$ has a natural Lie group structure.
The extensions $j_{\aa\bb}$ admit a quotient by scalings
and $\Del : X^{op} \aw \Grp$ also defines a contravariant functor.
The Gibbs state associated to any observable $U_\aa \in \obs_\aa$ 
is the non-vanishing probability density $[\e^{-U_\aa}] \in \Del_\aa$
and we denote by:
\[ \gibbs_\aa : G_\aa \aw \Del_\aa \]
the surjective group morphism induced by Gibbs states.

The convex set $\Delta(X) \incl \obs^*(X)$ of statistical fields 
may be defined by inclusion of each $\Delta_\aa \incl \obs^*_\aa$
and we denote by $\Del(X)$ the interior of $\Delta(X)$. 
We say that $p \in \Delta(X)$ is a consistent statistical field if $dp = 0$ 
and denote by $\Zp \incl \Del_0(X)$ 
the subset of non-vanishing consistent statistical fields.

\subsection{Homology}

\subsubsection{Gauss Formula.}

For every region $\aa \in X$, let us define the coboundary 
of the subsystem $\Lambda^\aa$ as the subset of arrows 
$\cobord{\aa} \incl N^1 X$ that meet $\Lambda^\aa$ but  
are not contained in $\Lambda^\aa$: 
\[ \cobord{\aa} = 
\{ \aa' \aw \bb' \st \aa' \not\in \Lambda^\aa \txt{and} \bb' \in \Lambda^\aa \}
\]
The following proposition may then be thought of as a Gauss formula on $\Lambda^\aa$: 

\begin{prop}
    For every $\ph \in \obs_1(X)$ and $\aa \in X$ we have:
    \[ \sum_{\bb' \in \Lambda^\aa} \bord_{\bb'} \ph = 
    \sum_{\aa'\bb' \in \cobord{\aa}} \ph_{\aa'\bb'} \]
    In particular, the above vanishes if $\ph$
    is supported in $\Lambda^\aa$.
\end{prop}

\begin{proof} In the sum of $\bord \ph$ over $\Lambda^\aa$, 
each term $\ph_{\bb\cc}$ appears with opposite signs 
in $\bord_{\bb}\ph$ and $\bord_{\cc} \ph$ 
whenever both $\bb$ and $\cc$ belong to $\Lambda^\aa$.
\end{proof}

We may write a similar formula on the cone $V_\bb$ over $\bb$ in $X$,
formed by all the regions that contain $\bb$,
and define its coboundary $\delta V_\bb$ as the set of arrows 
leaving $V_\bb$.
The sums however need to be embedded in the space of global observables.

\begin{prop}
    For every $\ph \in \obs_1(X)$ and $\bb \in X$ we have:
    \[ \sum_{\aa' \in V_\bb} \bord_{\aa'}\ph = 
    \sum_{\aa'\bb' \in \delta V_\bb} \ph_{\aa'\bb'} \]
    as global observables of $\obs_\Om$. 
\end{prop}

\subsubsection{Interaction Decomposition.}

For each region $\aa \in X$, 
We call boundary observable on a region $\aa \in X$ any observable 
generated by observables on strict subregions of $\aa$. 
We denote by $\b_\aa$ the subspace of boundary observables:
\[ \b_\aa = \sum_{\aa > \bb'} \obs_{\bb'} \]
Suppose chosen for every $\aa$ a supplement $\z_\aa$ of boundary observables, 
called interaction subspace of $\aa$,
so that we decompose $\obs_\aa$ as:
\[ \obs_\aa = \z_\aa \oplus \b_\aa \]
We may inductively continue this procedure on $\b_\aa$,
as illustrated by the following well known\footnote{
    The first appearance of this now very common result in statistics
    seems to be in Kellerer \cite{Kellerer-64}. 
    See also \cite{Matus-88} for a proof via harmonic analysis
}
theorem. 

\begin{thm}[Interaction Decomposition]
Given supplements $(\z_\aa)$ of boundary observables 
for every $\aa \in X$, we have the decompositions:
\[ \obs_\aa = \bigoplus_{\aa \aw \bb'} \z_{\bb'} \]
\end{thm}

The interaction decomposition induces a family of coherent projectors 
$P^{\bb\aa}$ of $\obs_\aa$ onto $\z_\bb \incl \obs_\aa$. 
They induce a map $P^\bb : \obs_0(X) \aw \z_\bb$
accounting for the contributions of all regions containing $\bb$:
\[ P^\bb(u) = \sum_{\aa' \aw \bb} P^{\bb\aa'}(u_{\aa'}) \]
Writing $\z_0(X) \incl \obs_0(X)$ for the direct sum of interaction subspaces, 
they define a projection:
\[ P : \obs_0(X) \aw \z_0(X)\]
We call $P$ the interaction decomposition. 

Given $u \in \obs_0(X)$, we may define a 
global observable $\zeta_\Om(u) \in \obs_\Om$ by: 
\[ \zeta_\Om(u) = \sum_{\aa \in X} u_\aa \]
The interaction decomposition essentially characterises the global sum 
of local observable fields.

\begin{corollary}
    For any $u \in \obs_0(X)$, we have the equivalence:
    \[ P(u) = 0 \quad \eqvl \quad \zeta_\Om(u) = 0 \]
    In particular, $\z_0(X)$ is isomorphic
    to the image of $\zeta_\Om$ in $\obs_\Om$\footnote{
        They both represent the inductive limit of $\obs$ over $X$.
    }. 
\end{corollary}

\begin{proof}
    When $X$ does not contain $\Om$, we have $\obs_\Om = \z_\Om \oplus \z_0(X)$,
    and the image of $\zeta_\Om$ may be viewed as the subspace of
    boundary observables on $\Om$.
\end{proof}

\subsubsection{Homology Groups.}

The complex of observable fields $\obs(X)$ is also acyclic, 
as might be expected\footnote{
We do not provide a proof here,
a treatment of higher degrees shall be given in later work.
}
from the retractability of $NX$, 
and we only focus on the first homology group.

\begin{thm} 
    The interaction decomposition $P$ induces 
    an isomorphism on the first homology group
    of observable fields:
    \[ \obs_0(X) / \bord \obs_1(X) \sim \z_0(X) \]
\end{thm}

\begin{proof}
    The Gauss formula on the cone $V_\bb$ above $\bb$ in $X$ 
    first ensures that $P$ vanishes on boundaries:
    \[ P^\bb(\bord \ph) = \sum_{\aa' \aw \bb} P^\bb(\bord_{\aa'} \ph) 
    = \sum_{\aa' \aw \bb} \sum_{\bb' \not\aw \bb} P^\bb(\ph_{\aa'\bb'}) 
    = 0 \]
    as $P^\bb(\obs_{\bb'})$ is non-zero if and only if $\bb'$ contains $\bb$.
    Let us denote by $[P]$ the quotient map induced by $P$. 
    Given $u \in \obs_0(X)$,
    consider the flux $\ph$ defined by
    $\ph_{\aa\bb} = P^{\bb\aa}(u_\aa)$:
    \[ \bord_\bb \ph = 
    \sum_{\aa' \aw \bb} \ph_{\aa'\bb} - \sum_{\bb \aw \cc'} \ph_{\bb\cc'} 
    = P^{\bb}(u) - u_\bb
    \]  
    When $P(u) = 0$ this gives $u = - \bord \ph$, hence $[P]$ is injective. 
\end{proof}

\begin{corollary} 
    Let $V = \zeta \cdot v$ denote an observable field of $\obs_0(X)$,
    and $c \in \R_0(X)$ denote the right Möbius numbers. 
    We have the equivalence:
    \[ c V \in \Img(\bord)
    \quad \eqvl \quad 
    v \in \Img(\bord) \]
\end{corollary}

\begin{proof}
    According to the theorem, it suffices to show that $P(v) = P(cV)$, as: 
    \[ P^\cc(v) = P^\cc(\mu \cdot V) = 
    \sum_{\aa'\aw \bb' \aw \cc} P^{\cc\bb'} (\mu_{\aa'\bb'} V_{\bb'}) 
    = \sum_{\bb' \aw \cc} P^{\cc\bb'} (c_{\bb'} V_{\bb'}) 
    = P^\cc(c V) \]
    for every $\cc \in X$.
\end{proof}

Theorem 2 and its corollary are at the core of our homological interpretation
for critical points of Bethe free energy and fixed points of Belief Propagation.

\section{First Applications}

\subsection{Critical Points of Bethe Free Energy}

\subsubsection{Gibbs Free Energy.}

We denote by $\Fg_\Om$ the global Gibbs free energy of $\Om$, 
viewed as the smooth functional on $\Delta_\Om \times \obs_\Om$ defined by:
\[ \Fg_\Om(p_\Om, H_\Om) = \E_{p_\Om}[H_\Om] - S(p_\Om) \]
where $S$ denotes Shannon Entropy.
Given a global hamiltonian $H_\Om \in \obs_\Om$, we denote by
$\Fg_\Om^{H_\Om}$ the induced functional on $\Delta_\Om$. 
The thermodynamic equilbrium state\footnote{
    When $\Om$ describes a canonical ensemble, {\it i.e.}
    a closed system interacting with a thermostat.
    The inverse temperature $(k_B T)^{-1}$ is set to $1$.
}
of $\Om$ is the global minimum of its Gibbs free energy $\Fg_\Om^{H_\Om}$.
Describing the cotangent fibers of $\Del_\Om$ by the 
quotient $\obs_\Om / \R$,
this defines a unique statistical state $p_\Om \in \Del_\Om$ by:
\[ \frac {\dr \Fg_\Om} {\dr p_\Om} \simeq
H_\Om + \ln(p_\Om) \simeq 0 \mod \R \] 
This means $p_\Om$ is the global Gibbs state $[\e^{-H_\Om}]$
associated to $H_\Om$. 

This definition of equilibrium being hardly computable in practice,
we shall seek to approximate the marginals 
$\Sigma^{\aa\Om}(p_\Om)$ of the global Gibbs state 
on the regions $\aa \in X$ of an admissible decomposition of $\Om$.

We first introduce the local Gibbs free energy $\Fg_\aa$ 
of a region $\aa \in X$ as the functional on $\Delta_\aa \times \obs_\aa$ 
defined by:
\[ \Fg_\aa(p_\aa, H_\aa) = \E_{p_\aa}[H_\aa] - S(p_\aa) \]
Given a local hamiltonian $H_\aa$ 
we denote by $\Fg^{H_\aa}_\aa$ the induced functional 
on $\Delta_\aa$. 
Its minimum is the local Gibbs state $[\e^{-H_\aa}]$ which would describe the  
equilibrium of $\aa$ had it been isolated from $\Omega \setminus \aa$.

\subsubsection{Bethe Approximation.}

The Bethe-Peierls approach and its refinements\footnote{
    Namely the Cluster Variation Method introduced by Kikuchi \cite{Kikuchi-51},
    whose variational free energy was related to a truncated
    Möbius inversion by Morita \cite{Morita-57}.
}
essentially consist in writing 
an approximate decomposition of $\Fg_\Om$ as a sum of 
local free energy summands:
\[ \Fb = \sum_{\bb \in X} \fg_{\bb} \]
where each $\fg_\bb$ is a smooth functional on $\Delta_\bb \times \obs_\bb$.
This localisation procedure can be made exact on any 
region $\aa \in X$, defining the free energy summands of $\Fg_\aa$
by Möbius inversion:
\[ \Fg_\aa = \sum_{\aa \aw \bb'} \fg_{\bb'} 
\quad \eqvl \quad
 \fg_\bb = \sum_{\bb \aw \cc'} \mu_{\bb \cc'} \; \Fg_{\cc'} \]
The approximation only comes when $\Omega$ is not in $X$,
we may then write the error $\Fg_\Om - \Fb$ 
as a global free energy summand $\fg_\Om$. 
From a physical point of view, we expect $\fg_\Om$ to be small 
when sufficiently large regions are taken in $X$ by extensivity
of the global Gibbs free energy\footnote{
    Schlijper \cite{Schlijper-83} proved this procedure convergent to
    the true free energy per lattice point for the infinite Ising 
    2D-model.
}.  

The functional $\Fb$ is thus
defined for every $p \in \Delta_0(X)$ and $H \in \obs_0(X)$ by:
\[ \Fb(p,H) = \sum_{\bb \in X} c_{\bb} \cdot \Fg_{\bb}(p_{\bb},H_{\bb}) \]
We call $\Fb$ the Bethe free energy. 
Given a reference hamiltonian field $H \in \obs_0(X)$,
we denote by $\Fb^H$ the induced functional on $\Delta_0(X)$.

\subsubsection{Critical Points.}

Because of the Möbius numbers $c_\bb$ appearing in its definition, 
the Bethe free energy $\Fb$ is no longer convex in general, 
and $\Fb^H$ might have a great multiplicity\footnote{
    For numerical studies see \cite{Weiss-97,Murphy-Weiss-99,Knoll-2017},
    A first mathematical proof of multiplicity is given by Bennequin
    in \cite{Bennequin-IEM}.
}
of consistent critical points 
in $\Zp$. 

Using theorem 2, we provide with a rigorous characterisation of 
the consistent statistical fields critical for $\Fb^H$.
We hope to give a better geometric understanding of these fields 
by showing that
they bear a homological relationship with the reference hamiltonian field $H$. 

\begin{thm}
    A non-vanishing consistent statistical field $p \in \Zp$
    is a critical point of the Bethe free energy $\Fb^H$ constrained
    to $\Zp$ if and only if there exists a flux $\ph \in \obs_1(X)$ such that: 
    \[ -\ln(p) \simeq H + \zeta \cdot \bord \ph \mod \R_0(X) \]
\end{thm}

\begin{proof}
    To describe the normalisation constraints, 
    we may look at the quotient $\obs_0(X) / \R_0(X)$ as the cotangent space
    of $\Del_0(X)$ at $p$, and write the differential of $\Fb^H$ as:
    \[ \frac {\dr \Fb} {\dr p} \simeq 
    \sum_{\bb \in X} c_\bb \big( H_\bb + \ln (p_\bb) \big) \mod \R_0(X) \] 
    The flux term comes as a collection of Lagrange multipliers
    for the consistency constraints.
    Whenever $p$ is a critical point, 
    the differential of $\Fb^H$ 
    vanishes on $\Ker(d) = \Img(\bord)^\perp$ and we have:
    \[ c\big( H + \ln(p)\big) \in \Img(\bord) + \R_0(X) \]
    The corollary of theorem 2 is crucial\footnote{
        The proof given in \cite{Yedidia-2005} is problematic
        when there exists $\bb$ such that $c_\bb = 0$.
    }
    to state that this implies:
    \[ H + \ln(p) \in \zeta \cdot \Img(\bord)  + \R_0(X) \]
    as $\zeta \cdot \R_0(X) = \R_0(X)$. 
\end{proof}

\subsection{Belief Propagation as a Transport Equation}

\subsubsection{Effective Energy.} 
For every $\aa \aw \bb$ in $X$, we call effective energy the smooth submersion
$\Fh^{\bb\aa}$ of $\obs_\aa$ onto $\obs_\bb$ defined by:
\[ \Fh^{\bb\aa}(U_\aa) = -\ln\big( \Sigma^{\bb\aa}(\e^{-U_\aa}) \big) \]
From a physical perspective, 
one may think of $\Fh^{\bb\aa}(U_\aa)(x_\bb)$ as the conditional free energy
of $\Lambda^\aa$ given $x_\bb$. 
Mathematically, we view effective energy as interlacing marginal 
projections with the diffeomorphisms between each $\obs_\cc$ and 
the orbit of the counting measure $1^\flat_\cc \in \obs^*_\cc$ under $G_\cc$:
\[ \Fh^{\bb\aa} = \mlog_\bb \circ \Sigma^{\bb\aa} \circ \expm_\aa \]
which implies functoriality in the category of smooth manifolds.
A weaker yet fundamental property is that Gibbs state maps
induce the following commutative diagram between 
effective energy and marginalisation:
\[ \Sigma^{\bb\aa} \circ \gibbs_\aa = \gibbs_\bb \circ \Fh^{\bb\aa} \]

Effective energy also induces a smooth functional 
$\nabF$ from $\obs_0(X)$ to $\obs_1(X)$, 
which we call effective gradient. It is defined by:
\[ \nabF(H)_{\aa\bb} = H_\bb - \Fh^{\bb\aa}(H_\aa) \]
Given a field of potentials $h$, we will be interested
in the effective gradient of its local
hamiltonian field $H = \zeta \cdot h$. 
Letting $\Phi = - \nabF \circ \zeta$ gives:
\[ \Phi_{\aa\bb}(h) = \Fh^{\bb\aa}
\bigg(  \sum_{\bb' \in \Lambda^\aa \setminus \Lambda^\bb} h_{\bb'} 
\bigg) \]
which is the effective contribution of $\Lambda^\aa \setminus \Lambda^\bb$
to the energy of $\Lambda^\bb$. 

\subsubsection{Belief Propagation.}

Given a field of potentials $h \in \obs_0(X)$,
consider the following transport equation with initial condition $h$:
\[ \dot u = \bord \Phi(u) \]
We denote by $\Xi = \bord \Phi$ the induced vector field on $\obs_0(X)$
and by $(\e^{t\Xi})$ its flow. 
In absence of normalisation\footnote{
    This is the case of the original algorithm on trees as introduced by Pearl
    in \cite{Pearl-88}.
}, we claim that Belief Propagation 
is equivalent
to the naive Euler scheme\footnote{ 
    BP is actually for $\tau = 1$,  
    a different time scale would appear as exponent in 
    the multiplicative formulation.
}
approximating the flow of $\Xi$ by: 
\[ \e^{n \tau \Xi} \simeq (1 + \tau \Xi)^n \]
The conjugated vector field $\Xi^\zeta = \zeta \circ \Xi \circ \mu$
on local hamiltonians induces the dynamic over beliefs through
the isomorphism of $\obs_0(X)$ with $G_0(X)$.

We believe this new perspective reveals 
a strong homological character of Belief Propagation.
Denoting by $T_h(\ph)$ the transport of $h$ by a flux 
$\ph \in \obs_1(X)$:
\[ T_h(\ph) = h + \bord \ph \]
it is clear that the potentials $u$ remain in the image of $T_h$.
This yields a maximal set of conserved quantities in light
of theorem 2.

\begin{thm}
Let $q \in G_0(X)^\N$ denote a sequence 
of belief fields obtained by iterating BP.
The following quantity remains constant:
    \[ q_\Om = \prod_{\aa \in X} (q_\aa)^{c_\aa} \]
\end{thm}

\begin{proof}
The fact that $u \in \Img(T_h)$ is equivalent to $P(u) = P(h)$,
itself equivalent to $U_\Omega = H_\Omega$ where:
\[ U_\Omega = \sum_{\aa \in X} u_\aa = \sum_{\bb \in X} c_\bb U_\bb\]
The beliefs are given by $q = \e^{-U}$. 
\end{proof}

Belief Propagation is usually viewed as a dynamic over
messages. 
Letting $u = h + \bord \ph$, we may recover the dynamic on potentials 
while keeping messages in memory by considering the evolution:
\[ \dot \ph = \Phi\big(T_h (\ph)\big) \]
We denote by $\Xi'_h = \Phi \circ T_h$ the induced vector field on $\obs_1(X)$. 
Although it converges on trees, 
this algorithm is generally divergent in presence of loops
and beliefs need to be normalised in order to attain
projective equilibria.

\subsubsection{Normalisation.}

Because the effective gradient $\nabF$ is additive along constants 
and both $\zeta$ and $\bord$ preserve scalar fields, we have:
\[ \Xi \big(u + \R_0(X) \big) \incl \Xi(u) + \R_0(X) \]
and $\Xi$ induces a vector field $[\Xi]$ on the space $\Del_0(X)$ of non-vanishing
statistical fields, identified with $\obs_0(X) / \R_0(X)$.
Dynamics over normalised beliefs $[\Xi^\zeta]$ and over messages 
$[\Xi'_h]$ are similarly
induced by $\Xi^\zeta$ and $\Xi'_h$.

Given an initial hamiltonian field $H$ and an initial flux $\ph^0$, 
let us denote by:
\[  U^t = \e^{t\Xi^\zeta} \cdot \; U^0 \]
the integral curve of $\Xi^\zeta$ originating from 
$U^0 = H + \zeta \cdot \bord \ph^0$.
The normalised flow is recovered by defining 
beliefs as $q = [\e^{-U}]$ and messages as $m = [\e^{-\ph}]$ 
which gives:
\[ -\ln(q) \simeq H + \zeta \cdot \bord \ph \mod \R_0(X) \]
In virtue of theorem 3, this implies that $q$ is a critical point of the Bethe free
energy $\Fb^H$ constrained to $\Zp$ 
if and only if $q$ is a consistent statistical field.

Considering all beliefs that may be obtained by any choice of messages,
we may define a subspace of admissible beliefs as:
\[ \Del_H = \left\{ [\e^{- U}] \st
 U \in H + \zeta \cdot \Img(\bord) \right\} 
 \incl \Del_0(X) \]
Following Yedidia {\it et. al.}, we might call any consistent 
$q \in \Del_H \cap \Zp$ a fixed point of Belief Propagation\footnote{
    This terminology is somewhat ambiguous as 
    it does not mean that $q$ may be obtained by iterating BP from $[\e^{-H}]$.
}.
With this terminology, we can rephrase their initial claim \cite{Yedidia-2005}:

\begin{thm}
    Fixing a reference hamiltonian field $H$,
    fixed points of Belief Propagation are in one to one correspondence
    with critical points of the Bethe free energy.
\end{thm}


\bibliographystyle{alpha}
\bibliography{biblio}

\end{document}